\documentclass[preprint,12pt]{elsarticle}
\usepackage{url}
\usepackage{amsthm}
\newtheorem{definition}{Definition}
\newtheorem{theorem}{Theorem}
\newtheorem{proposition}{Proposition}
\newtheorem{lemma}{Lemma}

\usepackage{amssymb}
\usepackage{amsmath}
\usepackage{physics}
\usepackage{physics}
\usepackage{tikz}
\usetikzlibrary{positioning}
\usetikzlibrary{arrows.meta}
\usetikzlibrary{arrows.meta, positioning}

\journal{Information Sciences}

\begin{document}

\begin{frontmatter}



\title{Foundations of Quantum Granular Computing with Effect-Based Granules, Algebraic Properties and Reference Architectures}


\author{Oscar Montiel Ross} 

\affiliation{organization={Instituto Politécnico Nacional - CITEDI},
            addressline={AV. Instituto Politécnico Nacional 1310}, 
            city={Tijuana},
            postcode={22435}, 
            state={Baja California},
            country={México}}

\begin{abstract}
This paper introduces a formal framework for Quantum Granular Computing (QGC), which extends classical granular computing --- including fuzzy, rough, and shadowed granules --- to the quantum regime. Quantum granules are modeled as effects on a finite-dimensional Hilbert space, making memberships Born probabilities and embedding granulation in the standard formalism of quantum information theory. We establish structural results for effect-based granules, including normalization, monotonicity, the emergence of Boolean islands from commuting families, refinement under Lüders updates, and evolution under quantum channels via the adjoint channel in the Heisenberg picture. We connect QGC with quantum detection and estimation theory by viewing Helstrom’s minimum-error measurement for binary state discrimination as a system of Helstrom-type decision granules, i.e., soft counterparts of Bayes-optimal decision regions. Building on this, we introduce Quantum Granular Decision Systems (QGDS) and three reference architectures: Measurement-Driven Granular Partitioning, Variational Effect Learning, and Hybrid Classical–Quantum pipelines. Case studies on qubit granulation, two-qubit parity, and Helstrom-style soft decisions show that QGC reproduces fuzzy-like graded memberships and smooth decision boundaries while exposing the role of non-commutativity and entanglement. The framework provides a unified, mathematically grounded basis for operator-valued granules in quantum information processing and intelligent systems.
\end{abstract}


\begin{graphicalabstract}
\begin{center}
  \begin{tikzpicture}[>=stealth,
  node distance=1.6cm and 2.4cm, 
  scale=0.9,
  every node/.style={transform shape}]

  \tikzset{block/.style={draw, rounded corners,
                         minimum width=2.8cm,
                         minimum height=1cm,
                         align=center}}

  \node[block, fill=gray!10] 
        (data) {Classical data};

  \node[block, fill=blue!10, right=2.0cm of data] 
        (gran) {Classical granulation\\(fuzzy, rough, etc.)};

  \node[block, fill=green!10, right=2.0cm of gran] 
        (encode) {Quantum encoding\\$|\psi(x)\rangle$};

  \node[block, fill=orange!10, below=1.2cm of encode] 
        (meas) {Quantum granular measurement\\$\{E_i\}\Rightarrow p_i$};

  \node[block, fill=red!10, left=2.0cm of meas] 
        (decision) {Decision rule\\$y = D(p_1,\dots,p_m)$};

  \node[block, fill=red!5, left=2.0cm of decision] 
        (output) {Decision outcome\\(class / action)};

  \node[circle, draw, fill=orange!20,
        inner sep=1pt, font=\scriptsize]
        (g1) at ([yshift=-0.9cm,xshift=-0.8cm]meas.south) {$E_1$};

  \node[circle, draw, fill=orange!20,
        inner sep=1pt, font=\scriptsize]
        (g2) at ([yshift=-0.9cm]meas.south) {$E_2$};

  \node[circle, draw, fill=orange!20,
        inner sep=1pt, font=\scriptsize]
        (g3) at ([yshift=-0.9cm,xshift=0.8cm]meas.south) {$E_3$};

  \node[font=\scriptsize, align=center]
        at ([yshift=-0.35cm]g2.south) {quantum granules};

  \draw[->, thick] (data.east) -- (gran.west);
  \draw[->, thick] (gran.east) -- (encode.west);
  \draw[->, thick] (encode.south) -- (meas.north);
  \draw[->, thick] (meas.west) -- (decision.east);
  \draw[->, thick] (decision.west) -- (output.east);

\end{tikzpicture}

\end{center}

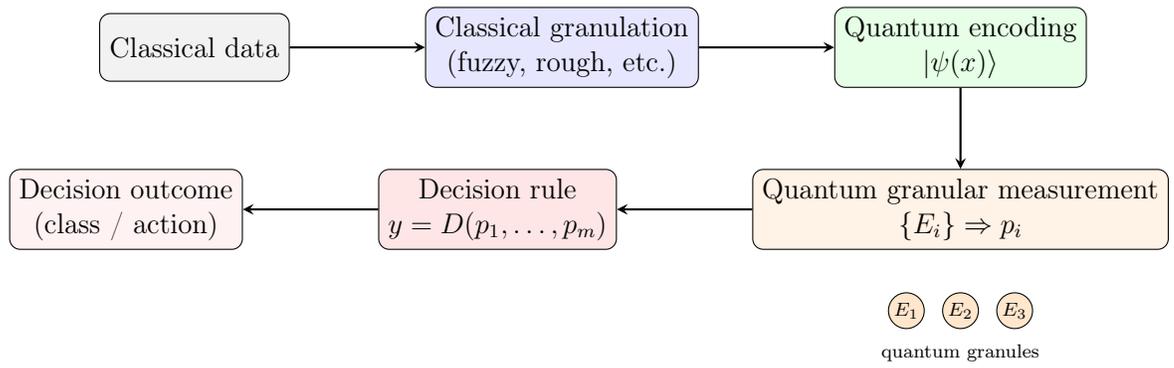
Classical--quantum granular decision pipeline underlying QGC.
\end{graphicalabstract}

\begin{highlights}
\item Introduce a framework for quantum versions of granular computing models.
\item Model information granules as quantum measurements acting on states.
\item Show how fuzzy and rough granules arise as classical limits of quantum ones.
\item Demonstrate quantum granular decision pipelines on simple qubit examples.
\end{highlights}


\end{frontmatter}


\section{Introduction}\label{sec:introduction}
Classical granular computing (GrC) provides a principled way to structure information through \emph{granules}—coherent units such as fuzzy sets, rough sets, interval-valued descriptions, and shadowed regions—that enable reasoning under uncertainty and abstraction \cite{zadeh1997,pawlak1982}. These models play a central role in intelligent systems and have been expanded through multilevel approximation schemes and structured granulation methodologies \cite{yao2004,yao1998}. Further expressive power has been achieved through shadowed sets, interval type-2 fuzzy sets \cite{pedrycz1998}, and Mediative Fuzzy Logic (MFL) \cite{Montiel2008MFL}, which allow granules to encode partial compatibility, conflict, and higher-order uncertainty. Nonetheless, classical granules typically operate within Boolean or fuzzy lattices, where all granules are assumed to coexist in a single distributive structure and can always be jointly combined. This global compatibility makes it difficult to represent phenomena driven by fundamentally incompatible observations, such as non-commuting measurements in quantum theory.

Quantum systems provide a natural setting in which these limitations become explicit. In quantum mechanics, states, measurements, and transformations are operators on a Hilbert space, and observable quantities generally do not commute \cite{nielsen2010,heinosaari2012}. This structure gives rise to \emph{contextuality}: many observables do not admit a joint description, and measurements can irreversibly disturb the system. At the same time, quantum information theory has shown that operator-valued objects such as effects, POVMs, and quantum channels offer a flexible language for soft decisions, uncertainty, and partial information \cite{busch2016,heinosaari2012}. The lattice of closed subspaces is orthomodular and non-distributive \cite{birkhoff1936}, and the set of effects forms an effect algebra \cite{foulis1994}, indicating that quantum theory already contains native algebraic structures suitable for granular modeling.

We therefore adopt the viewpoint that \emph{quantum granules} should be modeled as \emph{effects}, i.e., positive operators bounded by the identity. Effects are standard generalized events in quantum information, providing soft acceptance regions parameterized by their spectrum. Their associated Born probabilities $p_\rho(E)$, interpreted operationally as quantum degrees of membership, provide graded acceptance values for states \cite{busch2016}. This operator-based formulation captures both crisp granules (projectors) and soft granules (non-projective effects), creating a continuum that mirrors and extends fuzzy membership functions and the boundary regions of rough sets. The effect-algebra perspective \cite{foulis1994} gives QGC a mathematically coherent substrate grounded in quantum logic and operator theory.

Although the literature contains quantum-inspired approaches in areas such as cognition, decision theory, and machine learning, to the best of our knowledge no existing framework unifies operator-based granules, their algebraic properties, and computational architectures for intelligent systems. Classical GrC focuses on partitions, coverings, and approximation operators \cite{yao2004,yao1998}, which cannot be straightforwardly generalized to non-commutative settings. Rough-set–based reasoning has been widely explored in intelligent systems \cite{pawlak1982,SKOWRON2025122078}; yet its underlying classical algebra remains distributive and cannot express contextuality. Recent developments emphasize granular and topological refinements \cite{Ferreyra2023Topological}, while effect algebras provide logical foundations without operational pipelines \cite{foulis1994}. These gaps motivate the development of \emph{Quantum Granular Computing} (QGC): a coherent framework in which granules are genuinely quantum objects, and non-commutativity is treated as an essential feature rather than a limitation.

\subsection*{Contributions}\label{subsec:contributions}
This work makes the following contributions:
\begin{enumerate}
    \item We introduce a formal definition of \emph{Quantum Granular Computing} (QGC), in which quantum granules are modeled as effects on a finite-dimensional Hilbert space, and their membership degrees follow the probabilistic structure induced by the Born rule.
    \item We develop a foundational framework for effect-based granules and collect a suite of theorems characterizing their behavior, including normalization and monotonicity, Boolean islands induced by commuting families, and the evolution of granules under measurement and quantum channels via L\"uders updates and the Heisenberg adjoint of Kraus channels.
    \item We show that core models of classical granular computing---including fuzzy and rough granules---arise as commutative and projective special cases of effect-based quantum granules, thereby embedding classical GrC into the QGC framework and clarifying the classical--to--quantum correspondence.
    \item We demonstrate that optimal binary quantum decisions admit a natural interpretation as an \emph{optimal decision granule} of Helstrom type, linking the concepts of quantum granularity to quantum detection and estimation theory.
    \item We propose three reference architectures for QGC: Measurement-Driven Granular Partitioning (MDGP), Variational Effect Learning (VEL), and Hybrid Classical--Quantum pipelines (HCQ), emphasizing their suitability for near-term quantum devices.
    \item We present compact case studies on qubit granulation and binary quantum decision tasks---together with a quantum granular decision-system perspective---that establish parallels with fuzzy-style decision schemes while remaining rooted in operator theory.
\end{enumerate}

Overall, this paper provides a mathematically grounded and computationally relevant foundation for QGC, bridging classical granular models, operator algebras, and contemporary quantum information science.

\subsection*{Paper organization}\label{subsec:paper_organization}
The rest of this paper is organized as follows. 
Section~\ref{sec:background} recalls the classical roots of granular computing (GrC), reviewing fuzzy, rough, interval-valued, and shadowed sets, together with mediative fuzzy logic as representative models of information granules. 
Section~\ref{sec:Preliminaries} introduces the quantum formalism and notation used throughout, covering finite-dimensional Hilbert spaces, density operators, projective and positive operator-valued measurements, quantum channels, and the Heisenberg picture. 
Section~\ref{sec:quantum-granules} defines quantum granules as effects, develops their basic algebraic and probabilistic properties, establishes results on Boolean islands for commuting families, and analyzes granular refinement under measurements and the evolution of granules under quantum channels. 
Section~\ref{sec:qgds} formalizes Quantum Granular Decision Systems (QGDS) as an operator-valued decision architecture. 
Section~\ref{sec:Models} presents three reference architectures for QGC—Measurement-Driven Granular Partitioning (MDGP), Variational Effect Learning (VEL), and Hybrid Classical--Quantum pipelines (HCQ). 
Section~\ref{sec:case-studies} illustrates the framework through compact case studies on qubit granulation, two-qubit parity effects, and Helstrom-style soft decisions. 
Finally, Section~\ref{sec:discussion} analyzes implications, limitations, and directions for further work, and Section~\ref{sec:conclusions} summarizes the main contributions of the paper.

\section{Background}\label{sec:background}
Classical granular computing (GrC) grew out of early work on rough sets and information granules, where information is represented and processed in terms of coherent ``granules'' rather than individual points \cite{pawlak1982,yao2004,yao1998}. In this view, granules summarize local similarity, indiscernibility, or tolerance relations, and support multilevel descriptions of data and knowledge. Zadeh’s notion of information granulation further emphasizes the role of fuzzy sets and linguistic labels as human-centric carriers of granular knowledge \cite{zadeh1997}.

A first family of models is based on \emph{fuzzy} and \emph{rough sets}. Fuzzy sets allow graded membership, capturing vagueness in terms of membership functions instead of crisp boundaries \cite{zadeh1965fuzzy,pedrycz2020introFuzzyGranules,bargiela2003granular}, while rough sets describe concepts via lower and upper approximations induced by an indiscernibility relation \cite{pawlak1982,yao2004,Luo2025TrilevelRoughGranular}. Both frameworks provide algebraic operations (intersection, union, complement) and approximation mechanisms that make them natural carriers of information granules in classification, clustering, and decision-making tasks.

More expressive models extend these ideas to represent higher-order uncertainty and explicit boundary regions. \emph{Shadowed sets} compress intermediate membership degrees into a three-region structure (accepted, rejected, and shadowed) that mirrors the behavior of many decision processes \cite{pedrycz1998,pedrycz2005}. Fuzzy information granulation emphasizes the role of structured collections of fuzzy sets (information granules) as basic units of human-centric reasoning and control \cite{zadeh1997}. \emph{Interval-valued} \cite{turksen1986interval,PETRY2022108887} and \emph{type-2 fuzzy sets} \cite{karnik2001operations,mendel2002type2} further enrich this picture by attaching uncertainty to membership grades themselves, leading to granules that encode both value and confidence, which are particularly useful when expert assessments or data sources are imprecise.

Logical and algebraic extensions of fuzzy models have been proposed to accommodate conflict and higher-order reasoning. \emph{Mediative Fuzzy Logic} (MFL) provides a framework in which mediative connectives and graded entailment allow one to capture compromise and conflict resolution between granules \cite{MontielCastillo2007,CastilloMelin2023,MelinCastillo2025}. More recent developments further integrate MFL with granular modeling and intelligent systems, emphasizing architectures where information granules are explicitly manipulated as computational objects \cite{CastilloMelin2023,sharma2021mediative,iancu2018heart}. \emph{Formal words} and related constructions \cite{Lin1999FormalWords,Lin2003FormalWords} offer additional algebraic tools to describe the composition and refinement of granules in structured domains, extending the expressive power of classical GrC beyond purely numerical memberships and providing a conceptual symbolic counterpart to operator-based constructions in the quantum setting.

From a structural standpoint, most classical GrC models rely on Boolean or fuzzy algebras and on refinement mechanisms that are globally compatible across all attributes or criteria \cite{yao2004,yao1998}. Partitions, coverings, and neighborhood systems are typically assumed to coexist within a single distributive or modular lattice, so that granules can be refined, intersected, or aggregated without encountering fundamental incompatibilities. This setting is adequate for many classical information systems, but it does not account for situations in which the observations themselves may be context-dependent or mutually incompatible.

In parallel, the rough-set and granular-computing communities have developed increasingly refined approximation schemes for decision-making, including three-way decisions and probabilistic rough models with adjustable granularity. In particular, probabilistic rough sets and their three-region decision structures (accept, defer, reject) have been systematically studied in the context of information systems and classification tasks \cite{yao2010threeway,yao2011superiority}. More recent work in granular rough sets and topological approaches further emphasizes the role of granulation and neighborhood structure in information science applications \cite{SKOWRON2025122078,Ferreyra2023Topological}. These approaches illustrate how granules and approximations can be tuned to balance error, indeterminacy, and decision cost. They provide a natural conceptual bridge from classical GrC to more general, context-sensitive forms of granulation.

While the present work is, to our knowledge, the first to develop a general operator-theoretic framework for Quantum Granular Computing, there are emerging contributions that combine granular ideas with quantum algorithms in more specialized settings. Yuan et al.\ introduce quantum granular-ball generation methods that exploit the structure of quantum computation to accelerate the construction of granular-ball summaries and propose a quantum $k$-nearest neighbors classifier (QGBkNN), showing substantial reductions in time complexity compared with classical granular-ball approaches for KNN classification \cite{Yuan2025QuantumGranularBalls}. Acampora and Vitiello present a hybrid granular–quantum genetic optimization scheme in which a classical granular-computing layer hierarchically reduces the continuous search space, while a quantum processor evolves the genetic population, leading to statistically significant improvements over state-of-the-art quantum evolutionary algorithms on benchmark functions under current hardware constraints \cite{Acampora2023QuantumGeneticGranular}. These studies highlight a growing interest in leveraging granularity within quantum computing at the algorithmic level, but they do not address the foundational problem of representing information granules themselves as quantum objects, which is the focus of the present paper.

In this paper, we adopt a complementary perspective: instead of working with set-theoretic granules in a classical universe, we model granules as operators acting on a Hilbert space. This shift is motivated by quantum information theory, where states, measurements, and transformations are inherently operator-valued and non-commutative. Within this setting, we reinterpret classical notions such as approximations, decision regions, and multilevel descriptions in terms of quantum effects, positive operator-valued measures (POVMs), and channels, preparing the ground for the Quantum Granular Computing framework developed in the subsequent sections.

\section{Preliminaries and Notation}\label{sec:Preliminaries}
We establish here the notation and basic tools used throughout the paper. 
We write $\mathcal{H}$ for a finite-dimensional complex Hilbert space, 
$\mathrm{L}(\mathcal{H})$ (also written $\mathcal{B}(\mathcal{H})$) for the algebra of linear operators on $\mathcal{H}$, and $I$ for the identity operator on $\mathcal{H}$.  
For operators $A,B\in\mathrm{L}(\mathcal{H})$, we write $A\succeq 0$ to mean that $A$ is positive semidefinite, and
we define $A\preceq B$ (the L\"owner order) whenever $B-A\succeq 0$. The adjoint of $X$ is $X^\dagger$, and $\mathrm{Tr}(\cdot)$ denotes the trace~\cite{nielsen2010,heinosaari2012,wilde2017}.

\subsection{States, measurements, and effects}
\label{subsec:states-effects}
The set of quantum states (pure or mixed) on $\mathcal{H}$ is
\begin{equation}\label{eq:state-set}
    \mathcal{D}(\mathcal{H})
    = \bigl\{\,\rho \in \mathrm{L}(\mathcal{H}) : \rho\succeq 0,\ \mathrm{Tr}(\rho)=1 \,\bigr\}.
\end{equation}
Elements of $\mathcal{D}(\mathcal{H})$ are called \emph{density operators} or \emph{density matrices}.  

A \emph{pure state} is an extreme point of $\mathcal{D}(\mathcal{H})$, equivalently, a rank-one projector of the form
\begin{equation}\label{eq:pure-state}
    \rho = \ket{\psi}\!\bra{\psi},
    \qquad
    \langle\psi|\psi\rangle = 1.
\end{equation}
States that are not pure are called \emph{mixed}.  
Any mixed state admits a convex decomposition
\begin{equation}\label{eq:convex-decomposition}
    \rho = \sum_k p_k \ket{\psi_k}\!\bra{\psi_k},
    \qquad
    p_k \ge 0,\ \sum_k p_k = 1,
\end{equation}
where $\{\ket{\psi_k}\}$ are unit vectors in $\mathcal{H}$.  
The decomposition~\eqref{eq:convex-decomposition} is, in general, not unique, and the vectors $\ket{\psi_k}$ need not be orthogonal.  

A \emph{projective measurement} (PVM) is given by a finite family of pairwise orthogonal projectors
$\{P_i\}_i$ with $\sum_i P_i = I$ (equivalently, $P_i P_j = \delta_{ij} P_i$).  
Given a state $\rho\in\mathcal{D}(\mathcal{H})$, the Born rule assigns outcome probabilities
\begin{equation}\label{eq:born-pvm}
    p_i = \mathrm{Tr}(\rho P_i).
\end{equation}
Projective measurements are often regarded as \emph{ideal, sharp} quantum tests, in the sense that each outcome is represented by a projector with eigenvalues in $\{0,1\}$.

More generally, a \emph{positive operator-valued measure} (POVM) is a finite family of operators
$\{E_i\}_i$ such that each $E_i$ is positive semidefinite and
\begin{equation}\label{eq:povm-def}
    0 \preceq E_i \preceq I,
    \qquad
    \sum_i E_i = I.
\end{equation}
The elements $E_i$ are called \emph{effects}, and the outcome probabilities are
\begin{equation}\label{eq:born-povm}
    p_i = \mathrm{Tr}(\rho E_i).
\end{equation}
POVMs describe generalized measurements that may incorporate noise, coarse graining, or
unresolved degrees of freedom~\cite{busch2016,holevo2011}.

In this paper, the central objects are precisely these effects.  
An \emph{effect} is any operator $E$ with
\begin{equation}\label{eq:effect-def}
    0 \preceq E \preceq I.
\end{equation}
We denote the family of all effects on $\mathcal{H}$ by $\mathrm{Eff}(\mathcal{H})$.  
Projectors are the \emph{sharp} effects (idempotent elements with $P^2=P$), whereas non-projective effects represent \emph{soft} or \emph{unsharp} quantum events.  
For a fixed state $\rho$, the Born probability
\begin{equation}\label{eq:born-membership}
    p_\rho(E) := \mathrm{Tr}(\rho E)
\end{equation}
will be interpreted as the \emph{granular membership} of $\rho$ in the quantum granule specified by $E$.

\subsection{Open-system dynamics, quantum channels, and Heisenberg adjoint}\label{subsec:open-systems}
Open quantum systems are modeled by \emph{quantum channels}.  
A quantum channel from $\mathcal{H}_\mathrm{in}$ to $\mathcal{H}_\mathrm{out}$ is a linear map
\begin{equation}
    \mathcal{E} : \mathrm{L}(\mathcal{H}_\mathrm{in}) \to \mathrm{L}(\mathcal{H}_\mathrm{out})
\end{equation}
that is completely positive and trace-preserving (CPTP).  
Complete positivity can be characterized in terms of the positivity of the
Choi matrix $J(\mathcal{E})$ associated with $\mathcal{E}$,
defined by
\begin{equation}
    J(\mathcal{E}) = (\mathcal{E}\otimes \mathrm{id})
    (|\Phi^+\rangle\langle\Phi^+|),
\end{equation}
where $\mathrm{id}$ denotes the identity map on $\mathrm{L}(\mathcal{H}_\mathrm{in})$
and $|\Phi^+\rangle = \tfrac{1}{\sqrt{d}}\sum_{i=1}^{d}
|i\rangle\otimes|i\rangle$, with $d=\dim(\mathcal{H}_\mathrm{in})$, is a maximally
entangled state on $\mathcal{H}_\mathrm{in}\otimes\mathcal{H}_\mathrm{in}$ for a fixed orthonormal basis $\{|i\rangle\}_{i=1}^d$ of $\mathcal{H}_\mathrm{in}$. 
Thus $J(\mathcal{E})$ is an operator on $\mathcal{H}_\mathrm{out}\otimes\mathcal{H}_\mathrm{in}$. 
Choi's theorem states that $\mathcal{E}$ is completely positive if and only if
$J(\mathcal{E}) \succeq 0$~\cite{Choi1975}.

Every such map admits an operator-sum (Kraus) representation~\cite{Kraus1983}
\begin{equation}\label{eq:kraus-representation}
    \mathcal{E}(\rho)
    \;=\; \sum_{k} K_k \,\rho\, K_k^\dagger,
    \qquad
    \sum_{k} K_k^\dagger K_k \,=\, I,
\end{equation}
for suitable Kraus operators $K_k : \mathcal{H}_\mathrm{in} \to \mathcal{H}_\mathrm{out}$~\cite{nielsen2010,heinosaari2012,wilde2017}.  
The completeness relation in~\eqref{eq:kraus-representation} guarantees
trace preservation, and the representation is not unique: different sets
$\{K_k\}$ related by unitary mixing define the same channel. 
For a modern and self-contained account of quantum channels, Kraus and
Stinespring representations, and the Schr\"odinger--Heisenberg duality,
we refer to Watrous~\cite{watrous2018}.

Physically, a channel $\mathcal{E}$ can be obtained by coupling the system to an environment, applying a global unitary, and tracing out the environment (Stinespring dilation).  
In the \emph{Schr\"odinger picture}, the dynamics is described by the evolution of states, while observables (or effects) are kept fixed: a state $\rho$ is mapped to $\mathcal{E}(\rho)$.  
In the \emph{Heisenberg picture}, the states are kept fixed and observables or effects evolve according to the adjoint (dual) channel
\begin{equation}\label{eq:heisenberg-adjoint}
    \mathcal{E}^\dagger : \mathrm{L}(\mathcal{H}_\mathrm{out}) \to \mathrm{L}(\mathcal{H}_\mathrm{in}),
    \qquad
    \mathcal{E}^\dagger(E) = \sum_k K_k^\dagger E K_k,
\end{equation}
characterized by the trace duality
\begin{equation}\label{eq:trace-duality-channel}
    \Tr\!\bigl[\mathcal{E}(\rho)\,E\bigr]
    \;=\;
    \Tr\!\bigl[\rho\,\mathcal{E}^\dagger(E)\bigr]
\end{equation}
for all states $\rho$ and effects (or observables) $E$.  
For a trace-preserving channel, $\mathcal{E}^\dagger$ is a completely positive unital map, i.e.,
\begin{equation}\label{eq:adjoint-unital}
    \mathcal{E}^\dagger(I) = I.
\end{equation}

In the granular setting developed in this paper, we use the adjoint channel $\mathcal{E}^\dagger$ to describe how quantum granules (effects) are transformed by noisy dynamics. This viewpoint is central for analyzing the evolution of granular memberships under open-system dynamics and quantum decision processes.

\subsection{Lattice of subspaces and orthomodularity}\label{subsec:lattice}
Projectors on $\mathcal{H}$ are closely related to closed subspaces: each projector $P$ is associated with its range 
$\operatorname{ran}(P) \subseteq \mathcal{H}$, and conversely, each closed subspace arises as the range of a unique orthogonal projector.  
Let $\mathcal{L}(\mathcal{H})$ denote the set of all closed subspaces of $\mathcal{H}$ (equivalently, the ranges of projectors).  

For projectors $P,Q$, we write $P\preceq Q$ when
\begin{equation}\label{eq:projector-order}
    \operatorname{ran}(P) \subseteq \operatorname{ran}(Q),
\end{equation}
equivalently, when $PQ=QP=P$. Meet and join in $\mathcal{L}(\mathcal{H})$ are defined by
\begin{align}
P \wedge Q &:= \Pi_{\operatorname{ran}(P)\cap \operatorname{ran}(Q)}, \label{eq:lattice-meet}\\[2pt]
P \vee Q   &:= \Pi_{\operatorname{span}(\operatorname{ran}(P)\cup \operatorname{ran}(Q))}, \label{eq:lattice-join}
\end{align}
with orthocomplement
\begin{equation}\label{eq:lattice-orthocomplement}
    P^\perp := I - P.
\end{equation}
Here $\Pi_{S}$ denotes the orthogonal projector onto the closed subspace $S$.

The structure $(\mathcal{L}(\mathcal{H}),\wedge,\vee,{}^\perp)$ is an \emph{orthomodular lattice}~\cite{birkhoff1936}.  
When a family of projectors commutes, the sublattice they generate is Boolean (distributive), and classical set-theoretic reasoning applies.  
In general, however, distributivity fails, reflecting the incompatibility of non-commuting quantum events.  
This non-distributive structure is one of the core motivations for developing a quantum extension of granular computing, as it indicates that genuinely quantum events cannot be captured within classical Boolean or fuzzy lattices and therefore call for an operator-based granular framework.

\subsection{Granulation, locality, and composition}\label{subsec:local-global}
Granular models often need to distinguish between \emph{local} and \emph{global} structure. 
Here, \emph{local} refers to granules that act on a single subsystem only, whereas \emph{global} granules act on the joint system and may encode genuinely nonclassical correlations. 
In the quantum setting, this distinction becomes particularly significant for composite systems.  
Let $\mathcal{H}=\mathcal{H}_A\otimes\mathcal{H}_B$ describe a bipartite system with subsystems $A$ and $B$.  
A state on the joint system is a density operator $\rho_{AB}\in\mathcal{D}(\mathcal{H}_A\otimes\mathcal{H}_B)$; its reduced states are
\begin{equation}\label{eq:partial-trace}
    \rho_A = \mathrm{Tr}_B(\rho_{AB}),
    \qquad
    \rho_B = \mathrm{Tr}_A(\rho_{AB}).
\end{equation}

A \emph{local} granule on subsystem $A$ is represented by an effect of the form
\begin{equation}\label{eq:local-granule}
    E_A \otimes I_B,
    \qquad E_A \in \mathrm{Eff}(\mathcal{H}_A),
\end{equation}
(and analogously $I_A\otimes E_B$ for subsystem $B$). 
Here $\mathrm{Eff}(\mathcal{H}_A)$ denotes the set of all effects on $\mathcal{H}_A$, i.e., all operators $E_A \in \mathrm{L}(\mathcal{H}_A)$ such that $0 \preceq E_A \preceq I_A$. 
For such a granule, the granular membership in state $\rho_{AB}$ factorizes as
\begin{equation}\label{eq:local-membership}
    p_{\rho_{AB}}(E_A\otimes I_B)
    = \mathrm{Tr}\bigl(\rho_{AB}\,E_A\otimes I_B\bigr)
    = \mathrm{Tr}(\rho_A E_A)
    = p_{\rho_A}(E_A),
\end{equation}
showing that local granules depend only on the corresponding reduced state.

By contrast, a \emph{global} granule is an arbitrary effect
\begin{equation}\label{eq:global-granule}
    E_{AB} \in \mathrm{Eff}(\mathcal{H}_A\otimes\mathcal{H}_B),
\end{equation}
which need not factorize as $E_A\otimes E_B$.  
Here $\mathrm{Eff}(\mathcal{H}_A\otimes\mathcal{H}_B)$ denotes the set of effects on the composite system, i.e., all operators $E_{AB}$ with $0 \preceq E_{AB} \preceq I_{AB}$, where $I_{AB}=I_A\otimes I_B$. 
For entangled states $\rho_{AB}$, one may have
\begin{equation}\label{eq:entangled-granule}
    p_{\rho_{AB}}(E_{AB}) = 1
\end{equation}
for some joint effect $E_{AB}$, while no product interpretation in terms of separate granules on $A$ and $B$ is possible.  
Such global granules encode joint structure that is not reducible to independent local components and are central to our analysis of contextuality and entanglement.

This distinction between local and global granules will play a central role in the Quantum Granular Decision Systems of Section~\ref{sec:qgds} and in the case studies on two-qubit systems in Section~\ref{sec:case-studies}.  
It also mirrors the classical distinction between marginal and joint information granules in structured data domains, now enriched by genuinely quantum correlations.

\section{Quantum Granules and Foundational Results}\label{sec:quantum-granules}
In classical granular computing, granules are typically defined via membership functions, equivalence relations, or interval-valued descriptions.  
In quantum systems, these structures must be generalized to account for operator-valued information, contextuality, and incompatibility.  
Building on the preliminaries in Sec.~\ref{sec:Preliminaries}, we introduce \emph{quantum granules} as effect-based granules and establish several basic results describing their algebraic, probabilistic, and dynamical behavior.  
This section provides the formal core of Quantum Granular Computing (QGC).

\subsection{Quantum granules}\label{subsec:quantum-granule-definition}
Effects are taken as the fundamental units of quantum granulation.  
In analogy with classical granular models, such as fuzzy sets, rough sets, shadowed sets, MFL granules, and formal words, where a granule is represented by a universe together with one or more membership-like maps, we adopt the following definition.

\begin{definition}[Quantum granule]\label{def:quantum-granule}
Let $\mathcal{H}$ be a finite-dimensional Hilbert space, let $\mathcal{D}(\mathcal{H})$ be the corresponding state space, and let $\mathrm{Eff}(\mathcal{H})$ be the set of effects on $\mathcal{H}$.  
A \emph{quantum granule} on $\mathcal{H}$ is the triple
\begin{equation}\label{eq:qg-triple}
    \mathcal{G} = \bigl(\mathcal{D}(\mathcal{H}), E, p_E\bigr),
\end{equation}
where
\begin{enumerate}
    \item $E\in\mathrm{Eff}(\mathcal{H})$ is an \emph{effect}, i.e.
    \begin{equation}\label{eq:granule-definition}
        0 \preceq E \preceq I,
    \end{equation}
    \item $p_E:\mathcal{D}(\mathcal{H})\to[0,1]$ is the \emph{granular membership function} associated with $E$, defined by
    \begin{equation}\label{eq:granule-probability}
        p_E(\rho) \;:=\; \Tr(\rho E),
        \qquad \rho\in\mathcal{D}(\mathcal{H}).
    \end{equation}
\end{enumerate}
Since $\mathcal{H}$ and $\mathcal{D}(\mathcal{H})$ are fixed throughout this work, we often abbreviate the triple by writing $\mathcal{G}=(E,p_E)$, and even identify the granule with its effect $E$, using the shorthand $p_\rho(E)$ for $p_E(\rho)$.
\end{definition}

Mathematically, each effect $E$ thus induces a fuzzy-set–like membership function $p_E:\mathcal{D}(\mathcal{H})\to[0,1]$ on the quantum state space, in direct analogy with a classical fuzzy granule $(X,\mu_A)$, where $X$ is the universe and $\mu_A$ assigns membership grades to its elements.  
Projectors $P=P^\dagger=P^2$ appear as \emph{sharp} granules, with spectrum $\{0,1\}$ and $p_P(\rho)$ acting as a crisp acceptance indicator, while general (non-projective) effects encode \emph{soft}, noisy, and context-dependent acceptance regions.  

Beyond the purely fuzzy case, other classical granular formalisms can be embedded into this effect-based perspective.  
Rough granules, described by lower and upper approximations $(A_L,A_U)$, correspond in the commuting projective case to pairs of projectors $P_L\preceq P_U$, and more generally to pairs of effects $E_L\preceq E_U$ that bound the membership region for a concept.  
Shadowed sets and three-way decisions can be realized by finite families of effects $\{E_{\mathrm{acc}},E_{\mathrm{rej}},E_{\mathrm{und}}\}$ forming a POVM, where each effect defines an acceptance, rejection, or deferment granule at the quantum level.  
MFL-style mediative granules, with positive and negative components $(\mu^+,\mu^-)$, can be modeled by pairs of effects $(E^+,E^-)$ acting on the same state space, so that $p_{E^+}(\rho)$ and $p_{E^-}(\rho)$ play the role of mediative truth values; in the non-commuting case, $[E^+,E^-]\neq 0$ naturally encodes conflict and contextuality.  
Lin's formal words, which represent structured granules built from ordered compositions of simpler units, admit a quantum counterpart as ordered tuples or products of effects (or of POVM elements), where non-commutativity records the order sensitivity of the underlying granular operations.  
More generally, any classical granular construct specified by a finite family of membership-like maps can be transported to the quantum setting by assigning an appropriate finite family of effects on $\mathcal{H}$ whose Born-rule memberships reproduce the intended granular semantics.

From the logical side, the set $\mathrm{Eff}(\mathcal{H})$ of effects, equipped with the partial sum $E\oplus F := E+F$ whenever $E+F\preceq I$, forms an \emph{effect algebra}~\cite{foulis1994}, which provides the algebraic substrate for our granular constructions.

We record an explicit expression for the membership of pure states.

\begin{lemma}[Pure-state form]\label{lem:pure-state-form}
Let $\mathcal{G}=(\mathcal{D}(\mathcal{H}),E,p_E)$ be a quantum granule on $\mathcal{H}$.  
For a pure state $\rho=\ket{\psi}\!\bra{\psi}$ with $\langle\psi|\psi\rangle=1$, the granular membership is
\begin{equation}\label{eq:qg-pure-state}
    p_E(\rho)
    =
    \bra{\psi}E\ket{\psi}.
\end{equation}
In particular, if $E=P$ is a projector, then
\begin{equation}\label{eq:qg-pure-state-projector}
    p_P(\rho)
    =
    \bra{\psi}P\ket{\psi}
    =
    \|P\ket{\psi}\|^2,
\end{equation}
i.e., the squared norm of the projection of $\ket{\psi}$ onto the subspace $\operatorname{ran}(P)$.
\end{lemma}

\begin{proof}
Substituting $\rho=\ket{\psi}\!\bra{\psi}$ into~\eqref{eq:granule-probability} and using cyclicity of the trace yields
\[
p_E(\rho)
= \Tr\bigl(\ket{\psi}\!\bra{\psi} E\bigr)
= \bra{\psi}E\ket{\psi}.
\]
If $E=P$ is a projector, then $P^2=P$ and
\[
\bra{\psi}P\ket{\psi}
= \langle P\psi|P\psi\rangle
= \|P\ket{\psi}\|^2.
\]
\end{proof}

To build geometric intuition, Figures~\ref{fig:single-qubit-pure} and~\ref{fig:single-qubit-mixed} illustrate how effects behave as granules for pure and mixed qubit states.
\begin{figure}[t]
    \centering
    \begin{tikzpicture}[scale=2, >=Latex]

\draw[thick] (0,0) circle (1); 
 \draw[->] (-1.2,0) -- (1.2,0) node[below] {$X$};
 \draw[->] (0,-1.2) -- (0,1.2) node[left] {$Z$};

  \fill (0,1) circle (0.03) node[above right] {$\ket{0}$};
  \fill (0,-1) circle (0.03) node[below right] {$\ket{1}$};

  \pgfmathsetmacro{\ang}{40} 
  \coordinate (psi) at ({sin(\ang)},{cos(\ang)}); 
  \draw[thick,->] (0,0) -- (psi);
  \fill (psi) circle (0.03) node[right=2pt] {$\ket{\psi}$};

  \draw[->] (0,0) ++(0,0.25) arc [start angle=90, end angle=50, radius=0.25];
  \node at (0.15,0.55) {$\theta=$};

   
  \node at (0.15,0.40) {$\ang ^\circ$};

  \node[align=left, anchor=west] at (1.30,0.50) {$\displaystyle P_0=\ket{0}\!\bra{0}$\\[2pt] $\displaystyle P_1=\ket{1}\!\bra{1}$};

  \draw[->, dashed, bend left=15] (psi) to (0,1);
  \node[fill=white, inner sep=1pt] at (1.7,1.1){$\;p_0=\langle\psi|P_0|\psi\rangle=\cos^2(\!\frac{\theta}{2})\;$};

  \draw[->,dashed, bend left=15] (psi) to (0,-1);
  \node[fill=white, inner sep=1pt] at (1.7,-1.10){$\;p_1=\langle\psi|P_1|\psi\rangle=\sin^2(\!\frac{\theta}{2})\;$};

\end{tikzpicture}
    \caption{Quantum granules for pure qubit states.  
    The Bloch-sphere color map represents the granular membership $p_\rho(E)$ for a fixed effect $E$.  
    Projective effects (rank-one projectors) yield crisp spherical caps on the Bloch sphere, whereas non-projective ones induce soft granular boundaries.}
    \label{fig:single-qubit-pure}
\end{figure}
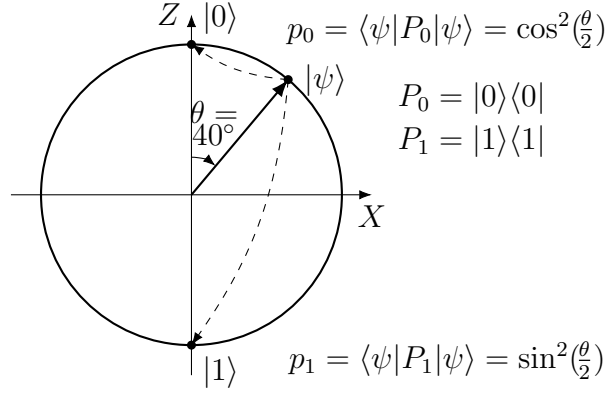

\begin{figure}[t]
    \centering
\newcommand{\BlochScale}{2}
\newcommand{\BlochR}{1.0}     

\begin{tikzpicture}[scale=\BlochScale, >=Latex]

  \def\rmag{0.75}   
  \def\marg{0.25}   
  \pgfmathsetmacro{\ang}{40} 

  \draw[line width=1.1pt] (0,0) circle (\BlochR);

  \draw[->,line width=0.7pt] (-\BlochR-\marg,0) -- (\BlochR+\marg,0) node[below right] {$X$};
  \draw[->,line width=0.9pt] (0,-\BlochR-\marg) -- (0,\BlochR+\marg) node[left] {$Z$};

  \fill (0,\BlochR) circle (0.03pt); \node[above right] at (0,\BlochR) {$\ket{0}$};
  \fill (0,-\BlochR) circle (0.03pt); \node[below right] at (0,-\BlochR) {$\ket{1}$};

  \coordinate (r)    at ({\rmag*\BlochR*sin(\ang)},{\rmag*\BlochR*cos(\ang)});
  \coordinate (rpure) at ({\BlochR*sin(\ang)},{\BlochR*cos(\ang)});
  \draw[line width=1pt] (0,0) -- (r) node[pos=0.58, right] {$\mathbf{r}$};
  \draw[dashed] (0,0) -- (rpure);

  \draw[densely dashed] (r) -- (0,{\rmag*\BlochR*cos(\ang)});
  \node at (-0.18,{\rmag*\BlochR*cos(\ang)/2}) {$r_z$};

  \draw[->] (0,0) ++(0,0.35) arc [start angle=90, end angle=50, radius=0.35];
  \node at (0.08,0.23) {$\theta$};
  \node at (0.15,0.45) {$\ang ^\circ$};

  \node[align=left] at (\BlochR+0.55,0.95) {$\rho=\tfrac12\!\left(I+\mathbf{r}\!\cdot\!\boldsymbol{\sigma}\right)$\\[2pt] $\|\mathbf{r}\|<1$};
  
  \node[align=left] at (\BlochR+0.90,0.50) {$p_0=\mathrm{Tr}(\rho P_0)=\tfrac{1+r_z}{2}$};

  \node[align=left] at (\BlochR+0.90,-0.50) {$p_1=\mathrm{Tr}(\rho P_1)=\tfrac{1-r_z}{2}$};
  
  \node[below right] at (0.10,\BlochR-0.05) {$p_0$};
  \node[above right] at (0.10,-\BlochR+0.05) {$p_1$};
\end{tikzpicture}
    \caption{Quantum granules for mixed qubit states.  
    As the Bloch-vector length $\|\vec{r}\|$ decreases (i.e., as the state becomes more mixed), granular contrasts soften, illustrating how uncertainty reduces the dynamic range of membership values.}
    \label{fig:single-qubit-mixed}
\end{figure}
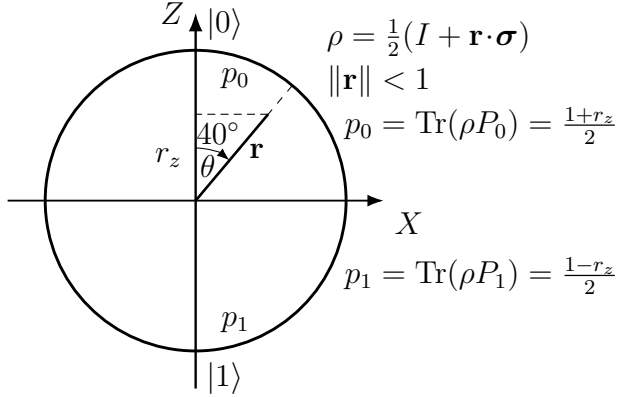

\subsection{Basic algebraic and probabilistic properties}\label{subsec:basic-properties}
We collect fundamental normalization and monotonicity properties of quantum granules, viewed through their membership functions.

\begin{theorem}[Normalization and monotonicity]\label{thm:normalization-monotonicity}
Let $\rho\in\mathcal{D}(\mathcal{H})$ and let $E,F\in\mathrm{Eff}(\mathcal{H})$ be effects (quantum granules) on $\mathcal{H}$.  
Then:
\begin{enumerate}
    \item $0 \le p_\rho(E) \le 1$,
    \item If $E \preceq F$, then $p_\rho(E) \le p_\rho(F)$,
    \item If $\{E_i\}_i$ is a POVM, then $\sum_i p_\rho(E_i)=1$.
\end{enumerate}
\end{theorem}

\begin{proof}
Since $0\preceq E\preceq I$ and $\rho\succeq 0$ with $\Tr(\rho)=1$, one has
\[
0 \;\le\; \Tr(\rho E) \;\le\; \Tr(\rho I) = 1,
\]
which proves~(1). For~(2), if $E\preceq F$ then $F-E\succeq 0$, and hence
\[
p_\rho(F)-p_\rho(E) = \Tr\!\bigl(\rho(F-E)\bigr) \;\ge\; 0,
\]
so $p_\rho(E)\le p_\rho(F)$. For~(3), linearity of the trace and $\sum_i E_i=I$ give
\[
\sum_i p_\rho(E_i)
= \sum_i \Tr(\rho E_i)
= \Tr\!\Bigl(\rho \sum_{i} E_i\Bigr)
= \Tr(\rho I)
= 1.
\]
\end{proof}

Equivalently, for each granule $\mathcal{G}=(\mathcal{D}(\mathcal{H}),E,p_E)$, Theorem~\ref{thm:normalization-monotonicity} states that $p_E$ is a normalized, order-preserving membership map on the state space $\mathcal{D}(\mathcal{H})$.

Theorem~\ref{thm:normalization-monotonicity} shows that quantum granules behave as generalized membership functions: 
their memberships are confined to $[0,1]$, respect the natural order on effects, and form a normalized probability distribution when arising from a POVM. 
In this sense, $p_\rho(E)$ plays the role of a fuzzy-style degree of membership, but induced by the Born rule on operator-valued granules.

Classical coarse-graining of partition blocks is recovered by coarse-graining commuting projective granules.

\begin{proposition}[Coarse-to-fine under commutation]\label{prop:coarse-to-fine}
Let $\{P_i\}$ be pairwise orthogonal, commuting projectors with $\sum_i P_i = I$ (a PVM).  
For any index set $S$, define
\begin{equation}\label{eq:qg-coarse-granule}
    P_S := \sum_{i\in S} P_i.
\end{equation}
Then $P_S$ is a projector that defines a coarser granule with membership
\begin{equation}\label{eq:qg-coarse-membership}
    p_{P_S}(\rho) = \sum_{i\in S} \Tr(\rho P_i),
\end{equation}
for all $\rho\in\mathcal{D}(\mathcal{H})$.
\end{proposition}

\begin{proof}
Orthogonality and commutativity of the $P_i$ imply that $P_S$ is a projector.  
Linearity of the trace gives
\[
p_{P_S}(\rho)
= \Tr(\rho P_S)
= \Tr\!\Bigl(\rho \sum_{i\in S} P_i\Bigr)
= \sum_{i\in S} \Tr(\rho P_i).
\]
\end{proof}

Proposition~\ref{prop:coarse-to-fine} generalizes the coarse-to-fine behavior of classical partitions: a union of blocks corresponds to a coarser granule whose membership function is the sum of the memberships of the constituent blocks.

\subsection{Boolean islands of commuting granules}\label{subsec:boolean-islands}

When a family of quantum granules commutes, it behaves classically.  
This provides a structural bridge between classical and quantum granulation, echoing the emergence of Boolean subalgebras in the lattice of subspaces~\cite{birkhoff1936}.

\begin{theorem}[Boolean islands]\label{thm:boolean-islands}
Let $\mathcal{E}=\{E_1,\dots,E_m\}\subset\mathrm{Eff}(\mathcal{H})$ be a finite family of pairwise commuting quantum granules.  
Then there exists a finite sample space $\Omega$, an algebra of subsets $\mathcal{F}$ of $\Omega$, and, for each state $\rho$, a probability measure $\mu_\rho$ on $(\Omega,\mathcal{F})$ such that:
\begin{enumerate}
    \item each $E_j$ is represented by a classical function $f_j:\Omega\to[0,1]$,
    \item $p_\rho(E_j)=\mu_\rho(f_j)$, where $\mu_\rho(f_j)=\sum_{\omega\in\Omega} f_j(\omega)\,\mu_\rho(\{\omega\})$,
    \item the projection lattice generated by the spectral projectors of the $E_j$ is (via the above identification) isomorphic to a Boolean algebra of subsets of $\Omega$.
\end{enumerate}
In the projective case, the $f_j$ reduce to characteristic functions $\mathbf{1}_{A_j}$ of events $A_j\in\mathcal{F}$, and $p_\rho(E_j)=\mu_\rho(A_j)$.
\end{theorem}

\begin{proof}[Proof sketch]
Let $\mathcal{A}$ be the unital $*$-subalgebra of $\mathrm{L}(\mathcal{H})$ generated by $\mathcal{E}$.  
Since the effects $E_j$ commute pairwise, every element of $\mathcal{A}$ can be simultaneously diagonalized in some orthonormal basis of $\mathcal{H}$.  
In that basis, $\mathcal{A}$ is identified with the algebra of diagonal matrices, which is $*$-isomorphic to an algebra of functions on a finite set $\Omega$~\cite{nielsen2010,heinosaari2012}.  
Each $E_j$ corresponds to a multiplication operator by a function $f_j:\Omega\to[0,1]$, and in the projective case, to a characteristic function $\mathbf{1}_{A_j}$ of some $A_j\subseteq\Omega$.  
The state $\rho$ induces a probability measure $\mu_\rho$ on $\Omega$, and $p_\rho(E_j)=\mu_\rho(f_j)$.  
The sets generated by the supports of the $f_j$ form a Boolean algebra of events.
\end{proof}

Viewed through QGC, commuting granules form \emph{Boolean islands} where classical granular reasoning is valid; incompatibility arises precisely when granules fail to commute.  
Within such an island, the granules $\{E_j\}$ can be identified with fuzzy-style sets $f_j:\Omega\to[0,1]$ under the classical probability measure $\mu_\rho$ (for a fixed state $\rho$), so that granular memberships and their combinations reduce to standard probabilistic reasoning on a Boolean algebra.

\subsection{Granular refinement under projective measurements}\label{subsec:luders-refinement}

The L\"uders measurement update for projective measurements describes how quantum states change upon observing an outcome. 
In our setting, it induces a decomposition of a granule's membership into conditioned components, providing a quantum granular analogue of classical conditioning.

\begin{theorem}[L\"uders granular refinement]\label{thm:luders-refinement}
Let $\{P_i\}_i$ be a projective measurement and let $\rho\in\mathcal{D}(\mathcal{H})$.  
Define $p_i=\Tr(\rho P_i)$ and, for $p_i>0$,
\begin{equation}\label{eq:luders-state}
    \rho_i = \frac{P_i \rho P_i}{p_i}.
\end{equation}
Let
\begin{equation}\label{eq:luders-nonselective}
    \rho' \;:=\; \sum_i P_i \rho P_i
           \;=\; \sum_i p_i\,\rho_i
\end{equation}
denote the post-measurement state corresponding to a nonselective L\"uders update.  
Then, for any effect $E\in\mathrm{Eff}(\mathcal{H})$ (quantum granule),
\begin{equation}\label{eq:luders-refinement}
    p_{\rho'}(E)
    \;=\;
    \sum_i p_i\, p_{\rho_i}(E).
\end{equation}
Moreover, if $E$ commutes with all $P_i$, this identity reduces to the classical law of total probability in the joint eigenbasis of $\{P_i\}$ and $E$~\cite{nielsen2010,heinosaari2012,busch2016}.
\end{theorem}

\begin{proof}
By definition of $\rho'$,
\[
p_{\rho'}(E)
= \Tr(\rho' E)
= \Tr\Bigl(\sum_i P_i \rho P_i E\Bigr)
= \sum_i \Tr(P_i \rho P_i E).
\]
For indices with $p_i=0$ the corresponding terms vanish and can be omitted.
For $p_i>0$ we have $p_i=\Tr(P_i\rho P_i)$ and
\[
\Tr(P_i \rho P_i E)
= p_i\,\Tr(\rho_i E)
= p_i\,p_{\rho_i}(E),
\]
so
\[
p_{\rho'}(E)
= \sum_i p_i\, p_{\rho_i}(E),
\]
which proves~\eqref{eq:luders-refinement}.  
In the commuting case $[E,P_i]=0$ for all $i$, $E$ and $\{P_i\}$ are simultaneously diagonalizable; in that basis,~\eqref{eq:luders-refinement} is exactly the classical law of total probability for the fuzzy set defined by the eigenvalues of $E$~\cite{nielsen2010,heinosaari2012,busch2016}.
\end{proof}

Theorem~\ref{thm:luders-refinement} shows that the membership of the post-measurement state $\rho'$ in a granule $E$ can be decomposed as a mixture of conditional memberships $p_{\rho_i}(E)$ given the outcomes of a PVM. 
This provides a quantum granular analogue of classical conditioning: after the measurement, $p_{\rho'}(E)$ is resolved into branches indexed by $i$, each weighted by $p_i$ and evaluated under the conditioned state $\rho_i$.

\subsection{Evolution of quantum granules under channels}\label{subsec:channel-evolution}

The Heisenberg adjoint of a quantum channel preserves granular structure and determines how granules transform under noisy dynamics~\cite{heinosaari2012,busch2016}.

\begin{theorem}[Granular evolution under quantum channels]\label{thm:channel-evolution}
Let $\mathcal{E} : \mathrm{L}(\mathcal{H}_{\mathrm{in}})\to\mathrm{L}(\mathcal{H}_{\mathrm{out}})$ be a quantum channel, and let $E\in\mathrm{Eff}(\mathcal{H}_{\mathrm{out}})$ be a quantum granule (effect) on the output space.  
Then:
\begin{enumerate}
    \item $\mathcal{E}^\dagger(E)$ is a granule on $\mathcal{H}_{\mathrm{in}}$, i.e., $\mathcal{E}^\dagger(E)\in\mathrm{Eff}(\mathcal{H}_{\mathrm{in}})$,
    \item If $E,F\in\mathrm{Eff}(\mathcal{H}_{\mathrm{out}})$ with $E\preceq F$, then $\mathcal{E}^\dagger(E)\preceq \mathcal{E}^\dagger(F)$,
    \item 3. For all $\rho\in\mathcal{D}(\mathcal{H}_{\mathrm{in}})$,
\begin{equation}\label{eq:channel-membership-duality}
    p_{\mathcal{E}(\rho)}(E) = p_\rho(\mathcal{E}^\dagger(E)).
\end{equation}

\end{enumerate}
\end{theorem}

\begin{proof}
(1) Since $\mathcal{E}$ is trace preserving, $\mathcal{E}^\dagger$ is completely positive and unital, i.e., $\mathcal{E}^\dagger(I_{\mathrm{out}})=I_{\mathrm{in}}$.  
Complete positivity of $\mathcal{E}^\dagger$ implies $\mathcal{E}^\dagger(E)\succeq 0$ whenever $E\succeq 0$, and unitality ensures $\mathcal{E}^\dagger(E)\preceq I_{\mathrm{in}}$ if $E\preceq I_{\mathrm{out}}$, so $\mathcal{E}^\dagger(E)\in\mathrm{Eff}(\mathcal{H}_{\mathrm{in}})$.

(2) If $E\preceq F$, then $F-E\succeq 0$ and complete positivity yields $\mathcal{E}^\dagger(F-E)\succeq 0$, i.e.,
\[
\mathcal{E}^\dagger(F)-\mathcal{E}^\dagger(E) = \mathcal{E}^\dagger(F-E)\succeq 0,
\]
so $\mathcal{E}^\dagger(E)\preceq \mathcal{E}^\dagger(F)$.

(3) Equation~\eqref{eq:channel-membership-duality} follows from Schr\"odinger–Heisenberg duality:
\[
p_{\mathcal{E}(\rho)}(E)
= \Tr\!\bigl(\mathcal{E}(\rho)\,E\bigr)
= \Tr\!\bigl(\rho\,\mathcal{E}^\dagger(E)\bigr)
= p_\rho\bigl(\mathcal{E}^\dagger(E)\bigr).
\]
\end{proof}

Theorem~\ref{thm:channel-evolution} shows that the effect of noise can be absorbed either into the state (Schr\"odinger picture) or into the granule (Heisenberg picture) without changing the resulting membership degrees.  
Operationally, measuring $E$ after a noisy evolution $\mathcal{E}$ is equivalent to measuring the \emph{dressed} granule $\tilde E=\mathcal{E}^\dagger(E)$ on the ideal state $\rho$.  
This perspective is important for NISQ-oriented implementations, where noise is unavoidable but can be modeled as a deformation of ideal granular structures.

\subsection{Optimal decision granules}\label{subsec:helstrom}

Binary quantum decisions admit a natural interpretation in terms of an \emph{optimal decision granule} described by Helstrom’s theorem~\cite{helstrom1969,holevo2011}.

\begin{theorem}[Optimal binary decision granule]\label{thm:helstrom}
Let $\rho_0,\rho_1$ be two states with priors $\pi_0,\pi_1>0$, $\pi_0+\pi_1=1$.  
Define
\begin{equation}\label{eq:delta-helstrom}
    \Delta = \pi_0 \rho_0 - \pi_1 \rho_1.
\end{equation}
Let $E^\star$ be the projector onto the positive spectrum of $\Delta$.  
Then:
\begin{enumerate}
    \item $E^\star$ is a (sharp) quantum granule,
    \item $E^\star$ maximizes the success probability
    \begin{equation}\label{eq:helstrom-success}
        P_{\mathrm{succ}}(E) 
        = \pi_0\, p_{\rho_0}(E) + \pi_1\, p_{\rho_1}(I-E),
    \end{equation}
    over all effects $E\in\mathrm{Eff}(\mathcal{H})$,
    \item The optimal value is
    \begin{equation}\label{eq:helstrom-opt}
        P_{\mathrm{succ}}(E^\star)
        = \tfrac{1}{2}\bigl(1 + \|\Delta\|_1 \bigr).
    \end{equation}
\end{enumerate}
\end{theorem}

\begin{proof}[Proof sketch]
Diagonalize $\Delta=\sum_j \lambda_j \Pi_j$ and expand $E$ in the same eigenbasis.  
The functional in~\eqref{eq:helstrom-success} can be written as
\[
P_{\mathrm{succ}}(E) 
= \tfrac12\bigl(1 + \Tr(\Delta(2E-I))\bigr),
\]
and is maximized by taking $E^\star=\sum_{\lambda_j>0}\Pi_j$, i.e., by accepting the positive eigenspaces of $\Delta$.  
Evaluating the optimum yields Eq.~\eqref{eq:helstrom-opt}; see~\cite{helstrom1969,holevo2011} for full details.
\end{proof}

From the QGC perspective, Theorem~\ref{thm:helstrom} shows that $E^\star$ is an optimal quantum granule that partitions the space into ``decide $\rho_0$'' vs.\ ``decide $\rho_1$'' according to Bayes risk.  
Any decision rule of the form $\{E,I-E\}$ can be viewed as a two-outcome POVM, where $E$ represents an acceptance granule for hypothesis $\rho_0$ and $I-E$ for $\rho_1$.  
The Helstrom granule $E^\star$ is optimal among all such effect-based rules, providing a canonical example of operator-valued granular decision-making and a template for fuzzy-like quantum decision schemes in QGC.

\section{Quantum Granular Decision Systems (QGDS)}\label{sec:qgds}

Quantum Granular Decision Systems (QGDS) provide an operational framework for decision-making in Quantum Granular Computing.  
A QGDS integrates effect-based granules, measurement-induced partitions, and classical aggregation rules into a single inference pipeline.  
It generalizes classical granular decision systems—such as fuzzy or MFL-based architectures—by replacing set-theoretic granules with quantum effects and by allowing decisions to depend on non-commuting operators, contextual information, and quantum state representations.

A QGDS operates on quantum states encoded from data and processes them through a prescribed sequence of granular evaluations.  
Each evaluation is specified by an effect (quantum granule), and the corresponding Born probabilities are interpreted as membership scores, which are subsequently aggregated into a final decision or inference.  
This construction extends classical soft-decision schemes to the operator domain and links them directly with optimality results from quantum detection theory (Sec.~\ref{subsec:helstrom}), where Helstrom-type decision granules provide canonical examples of Bayes-optimal quantum decisions.

\subsection{Structure of a QGDS}\label{subsec:qgds-structure}

A QGDS consists of four stages:

\begin{enumerate}
    \item \textbf{Classical or quantum input preprocessing.}  
    The input $x$ may be classical (numerical vectors, images, linguistic values) or already quantum (states prepared by upstream quantum processes).  
    Classical data can be preprocessed via fuzzy sets, rough approximations, or MFL-based mediative granules to obtain structured, interpretable features.

    \item \textbf{Quantum encoding.}  
    The (preprocessed) input is embedded into a state $\rho(x)$ on a suitable Hilbert space $\mathcal{H}$.  
    The encoding may be amplitude-based, density-operator–based, or hybrid (e.g., encoding fuzzy memberships as populations in $\rho(x)$).

    \item \textbf{Granular evaluation via effects.}  
    The system uses a family of effects $\{E_j\}_{j=1}^m$, typically arising from a PVM or, more generally, from a POVM.  
    The output of this stage is the vector of granular probabilities
    \begin{equation}\label{eq:qgds-memberships}
        p_j(x) = \Tr\!\bigl(\rho(x)\,E_j\bigr), \qquad j=1,\dots,m.
    \end{equation}
    These values are interpreted as soft quantum memberships, analogous to fuzzy degrees of membership but grounded in operator theory.

    \item \textbf{Aggregation and decision.}  
    A decision rule $D(\cdot)$ maps the granular probability vector to an output class, score, or action:
    \begin{equation}\label{eq:qgds-decision}
        y = D\bigl(p_1(x),\dots,p_m(x)\bigr).
    \end{equation}
    Examples include margin-based rules, Bayesian optimal rules (e.g., Helstrom-type decision granules), or fuzzy-style aggregations.
\end{enumerate}

The first two stages implement a data-to-state encoding, the third stage realizes effect-based granular evaluation as developed in Sec.~\ref{sec:quantum-granules}, and the fourth stage applies a classical decision layer.  
This architecture matches the hybrid pipeline illustrated in Fig.~\ref{fig:qgds_pipeline}, and provides a precise operational interpretation grounded in quantum granulation and its connection to optimal quantum decision rules (Sec.~\ref{subsec:helstrom}).

\subsection{QGDS and decision optimality}\label{subsec:qgds-optimality}

A key advantage of the QGDS framework is its compatibility with quantum-optimal decision granules.  
For binary decisions, Theorem~\ref{thm:helstrom} guarantees that the optimal granule is the projector onto the positive part of
\begin{equation}\label{eq:qgds-delta}
    \Delta = \pi_0 \rho_0 - \pi_1 \rho_1,
\end{equation}
yielding a pair of soft decision memberships
\begin{equation}\label{eq:qgds-soft-memberships}
    \mu_0(\rho) := p_{\rho}(E^\star),
    \qquad 
    \mu_1(\rho) := 1 - p_{\rho}(E^\star),
\end{equation}
where $E^\star$ is the Helstrom granule.  
In particular, QGDS recovers fuzzy-like soft classification while retaining Bayes-optimality properties from quantum detection theory.

Beyond binary tasks, POVMs naturally generalize granulation to multi-class settings.  
Any POVM $\{E_j\}_{j=1}^m$ defines a QGDS decision mechanism, with memberships $p_j(\rho)=\Tr(\rho E_j)$, and variational or data-driven methods (such as VEL in Sec.~\ref{subsec:vel}) can be used to learn granules optimized for specific objectives, subject to the partition constraint $\sum_j E_j = I$ and the effect bounds $0 \preceq E_j \preceq I$.

In the QGC framework, Theorem~\ref{thm:helstrom} provides a canonical example in which the granular decision rule is both soft and interpretable (via the memberships $\mu_0,\mu_1$) and provably optimal among all effect-based binary decisions.  
For multi-class problems, learned POVM-based QGDS architectures extend this principle by treating each effect $E_j$ as a class-specific decision granule, whose memberships are optimized directly in the operator domain.

\subsection{QGDS pseudocode}\label{subsec:qgds-pseudocode}

The following pseudocode summarizes the hierarchical structure of a QGDS and is compatible with both simulated systems and noisy intermediate-scale quantum (NISQ) devices. For clarity, the input may be either a classical datum $x$ or an already prepared quantum state $\rho_{\mathrm{in}}$; after Step~2, both cases are represented by a quantum state $\rho(x)$, and the subsequent stages are identical.

\begin{center}
\begin{minipage}{0.95\linewidth}
\hrule\vspace{0.4em}
\textbf{Algorithm: Quantum Granular Decision System (QGDS)}\\[0.3em]

\textbf{Input:} either a classical datum $x$ or a prepared quantum state $\rho_{\mathrm{in}}$\\
\textbf{Output:} decision $y$\\[0.4em]

\textbf{1. Classical preprocessing (only if the input is classical):}\\
\quad \textbf{if} the input is a classical datum $x$ \textbf{then}\\
\quad\quad compute classical granules $G_1,\dots,G_k$ with memberships $\mu_i(x)$\\
\quad\quad (e.g., fuzzy, rough, or MFL-based).\\[0.3em]

\textbf{2. State preparation / quantum encoding:}\\
\quad \textbf{if} the input is classical \textbf{then}\\
\quad\quad map $x$ or $\{\mu_i(x)\}$ to a quantum state $\rho(x)\in D(H)$;\\
\quad \textbf{else} (the input is already quantum) \textbf{set} $\rho(x) := \rho_{\mathrm{in}}$.\\[0.3em]

\textbf{3. Quantum granular evaluation:}\\
\quad Choose a PVM/POVM $\{E_j\}_{j=1}^m$ implementing the desired quantum granules.\\
\quad Measure (or estimate via repeated runs) granular memberships\\
\quad $p_j(x) = \Tr\!\bigl(\rho(x)\,E_j\bigr)$ for $j=1,\dots,m$.\\[0.3em]

\textbf{4. Classical aggregation and decision:}\\
\quad Compute $y = D\bigl(p_1(x),\dots,p_m(x)\bigr)$ for a chosen rule $D$\\
\quad (e.g., Helstrom-type rule, margin-based rule, or fuzzy-style aggregation).\\[0.3em]
\hrule
\end{minipage}
\end{center}

The QGDS algorithm can be interpreted as follows.  
When the input is classical (for example, a feature vector, an image, or a linguistic description), Step~1 optionally constructs classical information granules $G_i$ (such as fuzzy clusters or MFL-based mediative granules) with memberships $\mu_i(x)$. These memberships provide an interpretable, granulation-based description of $x$ that may be fed into the quantum part of the pipeline.  

In Step~2, the system prepares the quantum state to be processed. If the input is classical, the raw datum $x$ or the granule-based features $\{\mu_i(x)\}$ are encoded into a quantum state $\rho(x)\in D(H)$ on a Hilbert space $H$, for instance via amplitude encoding or density-operator encoding. If the input is already a prepared quantum state $\rho_{\mathrm{in}}$ (as in a communication or sensing scenario), Step~1 is omitted and Step~2 simply sets $\rho(x):=\rho_{\mathrm{in}}$. Thus, from Step~3 onwards the classical-input and quantum-input modes share the same quantum granular evaluation and decision pipeline.

Step~3 implements the core quantum granular evaluation. A PVM or POVM $\{E_j\}_{j=1}^m$ is chosen to represent the quantum granules relevant to the task (for example, Helstrom-type decision granules in the binary case or a learned POVM in a multi-class setting). Running the quantum device and measuring (typically with repeated shots on NISQ hardware) yields empirical estimates of the granular memberships $p_j(x)=\Tr\!\bigl(\rho(x)E_j\bigr)$, which act as soft, fuzzy-like degrees of membership in each decision granule.  

Finally, Step~4 applies a classical aggregation rule $D$, ranging from a simple $\arg\max_j p_j(x)$ or a margin-based rule to a Bayes-optimal decision rule inspired by Theorem~\ref{thm:helstrom}. The output $y$ is obtained by combining quantum granular memberships through this classical decision layer, completing the QGDS pipeline.

\subsection{Interpretability and granule hierarchy}\label{subsec:qgds-interpretability}

A notable feature of QGDS is the preservation of interpretability across quantum and classical stages.  
Classical granules (fuzzy sets, shadowed sets, MFL-based mediative granules) may feed into quantum granules through the encoding step, while quantum granular outputs $p_j(x)$ may be passed to downstream fuzzy or decision-theoretic modules.  
This hierarchical granulation mirrors multi-level approximation structures in classical GrC, but is now enriched by non-commutative operations, contextual quantum information, and effect-based representations of uncertainty.

Through this synthesis, QGDS provides a coherent computational model connecting operator-valued granules, quantum measurements, and classical inference.  
It also underpins the case studies in Sec.~\ref{sec:case-studies} and offers a template for architectures that integrate Quantum Granular Computing with broader quantum machine learning and decision-support systems.

\section{Models and Architectures}\label{sec:Models}
The definition of quantum granules provides the foundation for building computational models that exploit granularity within the quantum domain.  
In this section, we distinguish two complementary modeling views: hybrid schemes that combine classical granulation with quantum encoding, and fully quantum schemes in which granularity is defined directly in Hilbert space.  
We then revisit the Quantum Granular Decision Systems (QGDS) perspective of Sec.~\ref{sec:qgds}, emphasizing its role as a unifying decision architecture across these models.

Within this framework, we focus on three reference architectures for QGC.  
The first is the \emph{Measurement-Driven Granular Partitioning (MDGP)} pattern, where classical granules are formed at the data level and subsequently lifted to quantum effects.  
The second is \emph{Variational Effect Learning (VEL)}, in which quantum granules are parameterized effects learned from data.  
The third is a \emph{Hybrid Classical--Quantum (HCQ)} pipeline that combines classical granulation with a QGDS decision layer.  
These architectures specify how quantum granules are defined or learned and how they interface with the QGDS layer introduced in Sec.~\ref{sec:qgds}.

\subsection{Hybrid models: classical granulation with quantum encoding}\label{subsec:hybrid}
Hybrid models use established techniques of classical GrC to preprocess data into granules, which are then encoded into quantum states for further manipulation.  
They are \emph{hybrid} in the sense that granulation and part of the reasoning remain explicitly classical (fuzzy, rough, shadowed, MFL-based, or other GrC formalisms), while the subsequent evaluation and decision are carried out in the quantum, operator-valued domain.  

A typical hybrid QGC pipeline comprises three stages:
\begin{enumerate}
    \item[(i)] \textbf{Granulation step.}
    Given a dataset $X$, a classical granulation procedure partitions or clusters $X$ into granules $\{G_1,\dots,G_m\}$, each associated with membership degrees $\mu_i(x)$ for $x\in X$.  
    These memberships capture interpretable, graded inclusion of data points in each granule and can be obtained, for example, from fuzzy clustering, rough approximations, shadowed sets, MFL-based mediative schemes, or related GrC constructions.

    \item[(ii)] \textbf{Quantum encoding.}
    Each granule (or each granule-based feature vector) is then encoded into a quantum state.  
    A common approach is amplitude encoding:
    \begin{align}
    \ket{\psi_G} &= \frac{1}{\sqrt{Z}} \sum_{x\in G} \mu(x)\,\ket{x}, \label{eq:amplitude-encoding}\\
    Z &= \sum_{x\in G} |\mu(x)|^2, \label{eq:amplitude-norm}
    \end{align}
    where $\mu(x)$ is a (possibly normalized) membership function, $Z$ is a normalizing constant, and $\ket{x}$ denotes a computational-basis state labeled by $x$.  
    Alternatively, angle or basis encoding can map membership degrees to rotation parameters in a quantum circuit.  
    The choice of encoding affects circuit depth, resource requirements, and noise sensitivity, and is typically tailored to the target application.

    \item[(iii)] \textbf{Effect-based evaluation and decision (QGDS layer).}
    In the third stage, the encoded state $\rho(x)$ is processed by a Quantum Granular Decision System, as described in Sec.~\ref{sec:qgds}.  
    A family of effects $\{E_j\}$ implements the relevant quantum granules, yielding granular memberships $p_j(x)=\Tr\bigl(\rho(x)E_j\bigr)$ that are then aggregated by a classical decision rule $D$.  
    This QGDS layer realizes the final decision in terms of operator-valued granules, links the hybrid model to Helstrom-type optimal decision granules, and preserves interpretability by exposing both the classical and quantum granular structures.
\end{enumerate}

This architecture allows the direct use of quantum algorithms (for instance, quantum principal component analysis, amplitude amplification, or variational classifiers) on structured representations produced by classical GrC, combining the interpretability of classical granules with the potential expressiveness and parallelism of quantum state manipulations.  
In the broader QGDS setting, the hybrid architecture uses classical granulation to shape the input space and quantum effect-based granules to implement the final decision layer, so that the overall model remains explicitly interpretable while being genuinely quantum-enhanced.

\subsection{Purely quantum models}\label{subsec:purely-quantum}

Purely quantum models define granules directly within Hilbert space, without relying on classical preprocessing or external membership functions.  
In this regime, granules correspond to projectors or, more generally, to effects $\{P_i\}$ or $\{E_i\}$ satisfying
$\sum_i P_i = I$ (PVM) or $\sum_i E_i = I$ (POVM), with $0 \preceq E_i \preceq I$, and data are interpreted solely as quantum states $\rho$ processed through granular measurements.  
This paradigm aligns with quantum machine-learning approaches in which the model is fully encoded in operator-valued transformations and measurements~\cite{biamonte2017,schuld2018}.  
From the QGC viewpoint, the basic building block is therefore a measurement-driven granular model in which decisions are derived directly from the statistics of a PVM or POVM.

\subsubsection{Granular measurement model.}
A quantum system is probed using a PVM or POVM, where each outcome corresponds to a quantum granule specified by $(P_i,p_i)$ or $(E_i,p_i)$:
\begin{equation}\label{eq:purely-PVM}
    p_i = \Tr(\rho P_i),
\end{equation}
or
\begin{equation}\label{eq:purely-POVM}
    p_i = \Tr(\rho E_i).
\end{equation}
Granular hierarchies arise naturally by grouping outcomes into coarser effects
\begin{equation}\label{eq:purely-aggregate}
    E_S = \sum_{i\in S} E_i,
\end{equation}
for any subset $S$ of outcomes.  
Since $\sum_i E_i = I$ and $0 \preceq E_i \preceq I$, each $E_S$ is again an effect, which mirrors multilevel approximations in classical GrC but now expressed through operator sums~\cite{heinosaari2012,busch2016}.  
These hierarchies provide the quantum analogue of classical refinement and coarsening of granules.

\subsubsection{Advantages and expressive power.}
This approach eliminates the need for external membership functions: degrees of belonging are intrinsically defined through Born probabilities.  
Moreover, entangled states naturally support \emph{joint} granules acting across subsystems, which enables the representation of correlations and contextual structures with no classical analogue.  
Such non-classical granularity underpins the expressive power of purely quantum models and is central to variational quantum architectures~\cite{cerezo2021vqa}.

\subsection{QGDS as a unifying decision layer}\label{subsec:QGDS}
The architectures above can be viewed through the lens of Quantum Granular Decision Systems (QGDS), as introduced in Sec.~\ref{sec:qgds}.  
A QGDS provides a unifying decision layer that can be instantiated on top of both hybrid models (Sec.~\ref{subsec:hybrid}) and purely quantum models (Sec.~\ref{subsec:purely-quantum}).

Recall from Sec.~\ref{sec:qgds} that a QGDS comprises: (i) a family of effect-based quantum granules $\{E_j\}$ arising from a PVM or POVM, (ii) an encoding map $x\mapsto \rho(x)$ from inputs to quantum states, (iii) granular memberships $p_j(x)=\Tr\bigl(\rho(x)\,E_j\bigr)$ obtained by measurement, and (iv) a classical decision rule $D$ acting on the membership vector.  
In a hybrid architecture, the encoding stage is preceded by classical granulation, whereas in a purely quantum architecture the inputs are already given as quantum states.  
In both cases, the QGDS layer is responsible for translating effect-based granular information into a final decision.

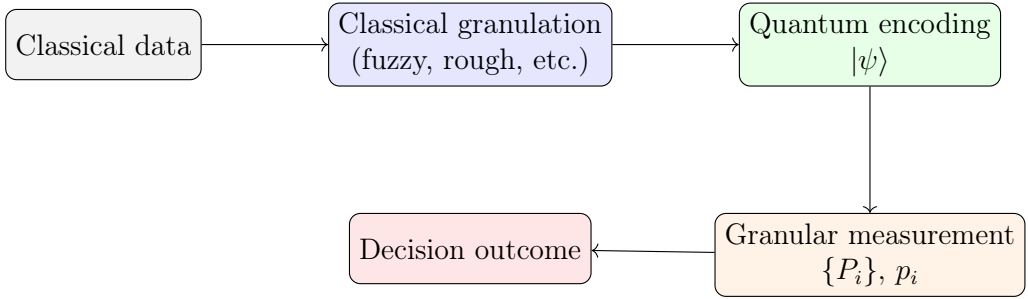
\begin{figure}[t]
  \centering
  \resizebox{\linewidth}{!}{\begin{tikzpicture}[node distance=1.8cm,scale=0.9, every node/.style={transform shape}]

\node[draw, rounded corners, fill=gray!10, minimum width=2.5cm, minimum height=1cm, align=center] (data)   {Classical data};

\node[draw, rounded corners, fill=blue!10, right=of data,        minimum width=3.0cm, minimum height=1cm, align=center] (gran)     {Classical granulation\\(fuzzy, rough, etc.)};
  
\node[draw, rounded corners, fill=green!10,  right=of gran,        minimum width=3.0cm, minimum height=1cm, align=center] (encode)   {Quantum encoding\\$\ket{\psi}$};

\node[draw, rounded corners, fill=orange!10, below=of encode,        minimum width=3.0cm, minimum height=1cm, align=center] (meas)     {Granular measurement\\$\{P_i\},\,p_i$};

\node[draw, rounded corners, fill=red!10, below=of gran,
      minimum width=2.6cm, minimum height=1cm, align=center] (decision) {Decision outcome};

  \draw[->] (data) -- (gran);
  \draw[->] (gran) -- (encode);
   \draw[->] (encode.south) -- (meas.north);   
  \draw[->] (meas) -- (decision);
          
\end{tikzpicture}}
  \caption{Hybrid pipeline for a Quantum Granular Decision System (QGDS).  
  Classical granulation structures the input, which is then encoded as a quantum state, processed through granular measurements (PVM/POVM), and mapped to a final decision by a Bayes-type rule acting on the resulting membership vector.}
  \label{fig:qgds_pipeline}
\end{figure}

\subsection{Variational Effect Learning (VEL)}\label{subsec:vel}

Measurement-Driven Granular Partitioning (MDGP) fixes the granular effects a priori, whereas Variational Effect Learning (VEL) treats them as trainable quantum parameters. In a VEL architecture, the QGDS granules are realized by a parametrized POVM $\{E_j(\theta)\}_{j=1}^m$ acting on the encoded state $\rho(x)$, and the parameters $\theta$ are optimized from data under the physical constraints
\begin{equation}
  0 \preceq E_j(\theta) \preceq I,
  \qquad
  \sum_{j=1}^m E_j(\theta) = I.
\end{equation}

A convenient parametrization is to start from a fixed POVM $\{F_j\}_{j=1}^m$ (for example, orthogonal projectors or block-diagonal templates) and conjugate it by a variational unitary $U(\theta)$:
\begin{equation}
  E_j(\theta)
  \;=\;
  U(\theta)^\dagger F_j U(\theta),
  \qquad j = 1,\dots,m.
  \label{eq:vel-effects}
\end{equation}
Because $\{F_j\}$ is a POVM and $U(\theta)$ is unitary, the family $\{E_j(\theta)\}$ is automatically a POVM for all $\theta$, i.e., $0 \preceq E_j(\theta) \preceq I$ and $\sum_j E_j(\theta) = I$ hold by construction. Structural priors such as sparsity, symmetry, or block structure can be encoded at the level of $\{F_j\}$ and are preserved by the conjugation.

Given an input $x$ encoded as a state $\rho(x)$, the granular probabilities in VEL are
\begin{equation}
  p_j(x;\theta)
  \;=\;
  \Tr\!\bigl(\rho(x)\,E_j(\theta)\bigr),
  \qquad j = 1,\dots,m,
  \label{eq:vel-probs}
\end{equation}
which enter directly into the generic QGDS decision rule
\begin{equation}
  y = D\bigl(p_1(x;\theta),\dots,p_m(x;\theta)\bigr)
\end{equation}
introduced in Sec.~\ref{sec:qgds}. For supervised tasks with labeled examples $\{(x^{(n)},y^{(n)})\}_{n=1}^N$, VEL optimizes $\theta$ to minimize a task-specific empirical risk
\begin{equation}
  \mathcal{L}(\theta)
  \;=\;
  \frac{1}{N}
  \sum_{n=1}^N
  \ell\Bigl(y^{(n)},\,p(x^{(n)};\theta)\Bigr),
  \label{eq:vel-loss}
\end{equation}
where $p(x^{(n)};\theta) = (p_1(x^{(n)};\theta),\dots,p_m(x^{(n)};\theta))$ and $\ell$ is a suitable loss function (for example, cross-entropy for classification or a margin-based loss). The optimization is carried out on a classical processor using gradient-based or gradient-free methods, while all constraints on $\{E_j(\theta)\}$ are enforced through the parametrization in~\eqref{eq:vel-effects}.

VEL thus provides a data-driven mechanism to learn quantum granules that remain physically valid effects and adapt to the structure of the task. In the commutative (classical) limit, where all $\rho(x)$ and $E_j(\theta)$ share a common eigenbasis, VEL reduces to learning classical soft granules over a Boolean island (Theorem~\ref{thm:boolean-islands}). In genuinely non-commutative settings, VEL can exploit Hilbert-space geometry and entanglement to realize granular decision boundaries with no direct classical analogue, while remaining compatible with NISQ hardware via shallow, structured ans\"atze for $U(\theta)$~\cite{cerezo2021vqa}.

\section{Case Studies}\label{sec:case-studies}

This section presents concise case studies illustrating how quantum granules behave in simple quantum systems.  
These examples connect the abstract definitions of effects and POVMs with concrete qubit-level behavior and with hybrid classical--quantum pipelines.  
We focus on one- and two-qubit systems because they provide clear geometric and algebraic intuition.

\subsection{Granulation of pure qubit states}\label{subsec:qubit-pure}

A pure qubit state can be written as
\begin{equation}\label{eq:pure-qubit}
\ket{\psi(\theta,\phi)}
= \cos\!\left(\tfrac{\theta}{2}\right)\ket{0}
+ e^{i\phi}\sin\!\left(\tfrac{\theta}{2}\right)\ket{1},
\end{equation}
with the density operator $\rho = \ket{\psi}\!\bra{\psi}$.  
Any effect $E$ defines a soft granule with membership probability
\begin{equation}\label{eq:pure-qubit-prob}
p_\rho(E) = \Tr(\rho E).
\end{equation}

If $E$ is projective, the granular region corresponds to a crisp spherical cap on the Bloch sphere.  
If $E$ is non-projective, the granular region is fuzzy, with $0 < p_\rho(E) < 1$ governed by the eigenvalues of $E$.

When the effects form a POVM $\{E_j\}$, the family of quantum granules $\{E_j\}$ induces a soft partition of the Bloch sphere in the QGC sense, with each $E_j$ acting as a quantum granule over the space of pure qubit states.

\subsection{Granulation of mixed qubit states}\label{subsec:qubit-mixed}

Mixed qubit states interpolate between pure states and the maximally mixed state $I/2$.  
A standard parametrization is
\begin{equation}\label{eq:mixed-qubit}
\rho
= \tfrac{1}{2}\bigl(I + r_x X + r_y Y + r_z Z\bigr),
\qquad \|\vec{r}\|\le 1,
\end{equation}
where $\vec{r} = (r_x,r_y,r_z)$ is the Bloch vector in the usual Bloch-sphere representation~\cite{nielsen2010}.

Any qubit effect can be written in Bloch form as
\begin{equation}\label{eq:effect-bloch}
E = \alpha I + \vec{e}\!\cdot\!\vec{\sigma},
\end{equation}
where $\alpha\in\mathbb{R}$, $\vec{e} = (e_x,e_y,e_z)\in\mathbb{R}^3$, $\vec{\sigma} = (X,Y,Z)$, and
\begin{equation}
\|\vec{e}\| = \sqrt{e_x^2 + e_y^2 + e_z^2}
\end{equation}
is the Euclidean norm of $\vec{e}$.  
For $E$ to be a valid effect (i.e., $0 \preceq E \preceq I$), the Bloch parameters must satisfy
\begin{equation}\label{eq:effect-constraints}
  0 \;\le\; \alpha \pm \|\vec{e}\| \;\le\; 1.
\end{equation}
The granular response is then
\begin{equation}\label{eq:mixed-response}
p_\rho(E)
= \alpha + \vec{r}\cdot\vec{e}.
\end{equation}

As $\|\vec{r}\|$ decreases, the contrast between high and low values of
$p_\rho(E)$ diminishes, reflecting the loss of discriminative power due
to mixing.

\subsection{Two-qubit granulation: parity effects}\label{subsec:two-qubit}

Two-qubit systems provide simple illustrations of commuting and non-commuting granules.  
Define the \emph{parity effects}
\begin{equation}\label{eq:parity-effects}
E_{\mathrm{even}} = \tfrac{1}{2}\bigl(I\!\otimes\!I + Z\!\otimes\!Z\bigr), 
\qquad
E_{\mathrm{odd}}  = \tfrac{1}{2}\bigl(I\!\otimes\!I - Z\!\otimes\!Z\bigr).
\end{equation}
These operators satisfy $E_{\mathrm{even}}+E_{\mathrm{odd}} = I$ and commute.  
Moreover, since $(Z\otimes Z)^2 = I\otimes I$, both $E_{\mathrm{even}}$ and $E_{\mathrm{odd}}$ are projectors, so they form a PVM and represent sharp, complementary granules corresponding to even and odd parity.

For any two-qubit state $\rho$,
\begin{equation}\label{eq:parity-probs}
p_{\mathrm{even}} = \Tr(\rho E_{\mathrm{even}}),
\qquad
p_{\mathrm{odd}} = 1 - p_{\mathrm{even}}.
\end{equation}
Thus, parity granulation reduces to classical binary classification on a Boolean island (Theorem~\ref{thm:boolean-islands}).  
This representation appears, for instance, in error detection and syndrome extraction, where only the global parity granule is relevant.

\subsection{Fuzzy-like decision via optimal granules}\label{subsec:fuzzy-like-decision}

Soft decision-making in QGC is naturally expressed through the optimal Helstrom granule (Theorem~\ref{thm:helstrom}).  
For states $\rho_0,\rho_1$ with priors $\pi_0,\pi_1$, recall the operator
\begin{equation}
\Delta = \pi_0 \rho_0 - \pi_1 \rho_1,
\end{equation}
as in Eq.~\eqref{eq:delta-helstrom}.  
The optimal decision granule $E^\star$ is the projector onto the positive part of $\Delta$.

For any decision granule $E$,
\begin{equation}\label{eq:helstrom-success-case}
P_{\mathrm{succ}}(E)
= \pi_0\,p_{\rho_0}(E) + \pi_1\,p_{\rho_1}(I-E),
\end{equation}
with optimal value
\begin{equation}\label{eq:helstrom-opt-case}
P_{\mathrm{succ}}(E^\star)
= \tfrac{1}{2}\bigl(1 + \|\Delta\|_1\bigr),
\end{equation}
where $\|\Delta\|_1$ denotes the trace norm of $\Delta$.

The induced soft decision rule is
\begin{equation}\label{eq:fuzzy-like-decision}
\mu_0(\rho) := p_{\rho}(E^\star),
\qquad
\mu_1(\rho) := 1 - p_{\rho}(E^\star).
\end{equation}
This provides a quantum, operator-theoretic analogue of fuzzy two-class decision schemes, with optimality grounded in quantum detection theory~\cite{helstrom1969,holevo2011}.

\subsection{Hybrid classical--quantum example}\label{subsec:hybrid-example}

Hybrid classical--quantum granular pipelines (HCQ) combine:
\begin{itemize}
    \item classical granules (e.g., fuzzy or rough),
    \item quantum encoding of classical memberships into $\rho(x)$,
    \item granular quantum queries (effects $E_j$),
    \item classical aggregation or decision-making.
\end{itemize}

In the QGDS framework, this pipeline follows the stages depicted schematically in Fig.~\ref{fig:qgds_pipeline}: classical data are granulated, encoded as quantum states, processed by quantum effects, and finally aggregated into a decision by a classical rule.

This setting illustrates how QGC integrates classical interpretability with quantum operator structure, enabling granular reasoning under incompatibility and contextuality while keeping the decision stage compatible with existing granular and fuzzy methodologies.

\section{Discussion}\label{sec:discussion}

The development of Quantum Granular Computing (QGC) provides an operator-theoretic extension of classical granulation that accommodates incompatibility, contextuality, and measurement-induced disturbance.  
This section discusses the interpretability of quantum granules, their relationship to classical granular models, and the implications of working in a non-distributive algebraic setting.  
It also outlines practical considerations for near-term implementations.

\subsection{Interpretability of quantum granules}\label{subsec:interpretability}

One of the goals of QGC is to preserve the interpretability that characterizes classical granular models.  
Quantum granules, defined as effects, retain a direct probabilistic meaning through $p_\rho(E)=\Tr(\rho E)$, which can be viewed as a degree of membership or compatibility between the state and the concept represented by $E$.  
In projective cases, granules correspond to crisp regions of Hilbert space, while non-projective effects yield soft or ``fuzzy-like'' boundaries.  

Despite being operator-valued, quantum granules support human-readable semantics when:
\begin{itemize}
    \item their eigenstructure corresponds to meaningful directions (for example, Bloch-sphere axes or physically interpretable observables),
    \item they arise from a physically motivated PVM/POVM (for example, energy levels, parity, or syndrome measurements),
    \item or they are learned through VEL with constraints promoting sparsity, symmetry, or block structure in $E$.
\end{itemize}
In this way, QGC allows interpretability comparable to fuzzy granulation, but grounded in the operational structure of quantum mechanics and measurement theory~\cite{heinosaari2012,busch2016,peres1995}.

\subsection{Relation to classical granulation models}\label{subsec:relation-classical}

QGC generalizes classical granules in multiple ways:
\begin{enumerate}
    \item \textbf{Fuzzy sets.} Soft effects $E$ provide operator-valued analogues of membership functions, where eigenvalues play the role of degrees of compatibility.  
    The Born probabilities $p_\rho(E)$ act as fuzzy-style memberships tied to the underlying state, connecting directly with classical fuzzy granulation~\cite{zadeh1997,pedrycz1998}.
    \item \textbf{Rough sets.} Projective decompositions $\{P_i\}$ form crisp approximations, and effect-valued granules $E_i = P_i E P_i$ capture boundary-like behavior between regions, reminiscent of lower/upper approximations and boundary regions in rough-set theory~\cite{pawlak1982,yao2004}.
    \item \textbf{Interval and shadowed granules.} The spectrum of $E$ induces an internal interval structure, while measurement disturbance naturally produces ``shadowed'' regions where membership cannot be sharpened without changing the state, echoing interval-valued and shadowed-set models~\cite{pedrycz1998,SKOWRON2025122078,Ferreyra2023Topological}.
    \item \textbf{MFL and mediative reasoning.} Mediative interactions in Mediative Fuzzy Logic correspond to compatibility and conflict between granules.  
    QGC extends this idea to operator-mediated interactions, via non-commuting effects and channels, that have no classical counterpart~\cite{MontielCastillo2007,CastilloMelin2023}.
\end{enumerate}

Quantum granules reduce to classical granules whenever the family of effects commutes.  
This phenomenon, formalized by the Boolean-islands result in Theorem~\ref{thm:boolean-islands}, provides a clear bridge between quantum and classical granular computation: on commuting islands, QGC collapses to a classical probability space; outside them, genuinely quantum features such as contextuality and entanglement emerge~\cite{birkhoff1936,peres1995}.

\subsection{Non-distributive effects and contextuality}\label{subsec:non-distributive}

A fundamental difference between QGC and classical GrC is the non-distributive nature of the underlying operator lattice.  
When effects do not commute, the corresponding granules cannot be jointly evaluated, nor can they be embedded into a common Boolean algebra.  
This non-distributivity reflects genuine contextuality: the outcome of querying one granule may depend on the preceding measurement, and there may be no single classical probability space that reproduces all granular memberships simultaneously~\cite{peres1995,heinosaari2012}.

Rather than being a limitation, this feature allows QGC to model phenomena such as:
\begin{itemize}
    \item mutually exclusive observations that cannot be realized in the same experimental context,
    \item incompatible measurement settings, where attempting to refine one granule degrades information about another,
    \item state disturbance induced by granule evaluation, which changes subsequent membership degrees in a controlled way (for example, through L\"uders refinement).
\end{itemize}
These phenomena have no classical analogue and represent a key advantage of adopting quantum structures for granular reasoning in domains where contextuality is intrinsic.

\subsection{Advantages for intelligent systems}\label{subsec:advantages}

QGC offers several potential benefits for intelligent systems:
\begin{itemize}
    \item \textbf{Soft decision boundaries.} Non-projective effects act as fuzzy-like granules, enabling smooth decision boundaries and graded memberships within a fully quantum formalism, and connecting naturally with quantum detection and estimation theory~\cite{helstrom1969,holevo2011}.
    \item \textbf{Context-aware granulation.} Measurement disturbance and channel evolution give rise to context-dependent granules whose memberships adapt to prior observations and noise models, reflecting the inherently contextual nature of quantum measurements~\cite{busch2016,heinosaari2012}.
    \item \textbf{Compatibility with quantum hardware.} Effects, POVMs, and channels are natively supported in quantum devices, allowing QGC-based models to be implemented within standard quantum-information toolkits and to interface with quantum machine-learning pipelines~\cite{biamonte2017,schuld2018}.
    \item \textbf{Hybrid integration.} QGC naturally supports classical--quantum pipelines in which fuzzy, rough, or MFL-based preprocessing feeds into quantum granules (MDGP, VEL, QGDS), preserving interpretability while exploiting Hilbert-space structure and entanglement for richer representations.
\end{itemize}
These features enable reasoning under uncertainty in regimes where classical granules cannot represent contextual or incompatible information, and they suggest a path toward explainable quantum models that remain compatible with existing GrC methodologies.

\subsection{Limitations and practical considerations}\label{subsec:limitations}

Several limitations and challenges remain for QGC, especially in near-term (NISQ) implementations:
\begin{itemize}
    \item \textbf{Noise sensitivity.} Non-projective effects can amplify sampling noise unless regularized.  
    Empirical memberships $\hat p_\rho(E)$ require sufficient shots and appropriate confidence intervals; noise and readout-mitigation routines must be integrated into the interpretation of granules~\cite{heinosaari2012,busch2016}.
    \item \textbf{Hardware restrictions.} Current devices constrain the depth and connectivity of ans\"atze used in VEL and QGDS.  
    This limits the complexity of learnable granules and may favor shallow, structured architectures over generic parametrizations, in line with observations on variational quantum algorithms~\cite{cerezo2021vqa}.
    \item \textbf{Scalability.} The number of parameters in effect-based learning can grow quadratically with the Hilbert-space dimension.  
    Exploiting symmetries, factorized structures, and low-rank or tensor-network parametrizations will be crucial for scaling to larger systems.
    \item \textbf{Interpretability of learned granules.} While structured ans\"atze improve readability, arbitrary learned effects may lack intuitive meaning.  
    Additional constraints (for example, sparsity, block structure, or commutation with physically meaningful observables) may be needed to maintain semantic transparency, particularly in decision-support applications where explainability is critical.
\end{itemize}

Despite these challenges, QGC provides a mathematically robust foundation for studying granulation in settings where quantum structure is intrinsic or advantageous.  
It complements classical GrC by offering a principled way to represent and manipulate granules in non-commutative spaces, and it opens a path toward transparent quantum models that reason with uncertainty, hierarchy, and context at the operator level.

\section{Conclusions}\label{sec:conclusions}

This paper has outlined an operator-theoretic framework for Quantum Granular Computing (QGC), extending classical granular computing, including fuzzy, rough, and shadowed granules, to the quantum regime. A central modeling choice is to represent quantum granules as effects on a finite-dimensional Hilbert space, so that granular memberships are given by Born probabilities. This operator-theoretic viewpoint offers a common language for sharp (projective) and soft (non-projective) granules and embeds granulation directly into the standard formalism of quantum information theory.

On this basis, we collected several basic properties of effect-based quantum granules. We characterized normalization and monotonicity properties, formulated a ``Boolean islands'' result for commuting families in which classical probabilistic reasoning is recovered, and analyzed granular refinement under Lüders updates. We also described the evolution of granules under quantum channels via the Heisenberg adjoint, thereby clarifying how noisy dynamics transform ideal granules into effective ones in non-commutative settings.

We also highlighted a connection between QGC and quantum detection and estimation theory by interpreting optimal binary decisions as Helstrom-type decision granules. In this formulation, the Helstrom projector provides a natural example of a quantum granule that realizes Bayes-optimal soft decisions, suggesting an analogy with fuzzy-style classifiers within an operator-theoretic framework. This link illustrates how granular reasoning and optimal quantum decision rules can be treated within a single formalism.

Building on these theoretical elements, we introduced Quantum Granular Decision Systems (QGDS) as an operational model for decision-making in QGC. We outlined three reference architectures: Measurement-Driven Granular Partitioning (MDGP), Variational Effect Learning (VEL), and Hybrid Classical--Quantum pipelines (HCQ). Together, these architectures illustrate possible ways in which quantum granules can be specified, learned, and integrated with classical components in a structured way, while remaining compatible with near-term quantum hardware and standard quantum-information primitives.

Compact case studies on qubit granulation, two-qubit parity effects, and Helstrom-style soft decisions served to illustrate how the proposed framework behaves in simple scenarios. These examples suggest that QGC can reproduce fuzzy-like behavior (soft granules, graded memberships, and smooth decision boundaries) while also exploiting non-commutativity, contextuality, and entanglement. In this sense, QGC does not simply mimic classical granulation in Hilbert space, but appears to enlarge the design space of granular models by allowing genuinely quantum phenomena.

Overall, we have proposed a mathematically grounded framework for QGC that connects classical granular models, operator theory, and modern quantum information science. In particular, the operator-based formulation provides a unified treatment of sharp and soft quantum granules while recovering fuzzy, rough, and shadowed granules as commutative and projective special cases of quantum granules, as formalized by the Boolean-islands result in Theorem~\ref{thm:boolean-islands}. At the level considered here, this supports the view that core models of classical GrC can be embedded as commutative, projective instances of QGC, and it provides a basis for further work on quantum information processing, granular reasoning, and intelligent systems in domains where contextuality and non-commutativity are intrinsic features of the data and the decision process.

\subsection*{Acknowledgments}
We gratefully acknowledge the support of the Secretaría de Ciencia, Humanidades, Tecnología e Innovación (SECIHTI), the Instituto Polit\'{e}cnico Nacional (IPN), and the Comisión de Fomento y Apoyo Acad\'{e}mico del IPN (COFAA), which made this research possible.


\end{document}